\documentclass[11pt, a4paper]{article}

\usepackage[margin=1in]{geometry} 
\usepackage{amsfonts, amscd, amssymb, mathtools, mathrsfs, dsfont, bbm, bbding} 
\usepackage[amsmath, amsthm, thmmarks]{ntheorem} 
\usepackage{graphicx, xypic, color, float} 
\usepackage{indentfirst}
\usepackage{lmodern}
\usepackage[T1]{fontenc} 
\usepackage{enumerate, listings, verbatim, paralist}
\usepackage{setspace, xspace}
\usepackage[colorlinks=true, citecolor=blue]{hyperref}
\usepackage{extarrows}
\usepackage{pdfpages}
\usepackage{ulem}

\usepackage{setspace}
\AtBeginDocument{%
\addtolength\abovedisplayskip{-0.15\baselineskip}%
\addtolength\belowdisplayskip{-0.15\baselineskip}%
\addtolength\abovedisplayshortskip{-0.15\baselineskip}%
\addtolength\belowdisplayshortskip{-0.15\baselineskip}%
}

\newtheorem{thm}{Theorem} [section]

\newtheorem{lemma}[thm]{Lemma}

\newtheorem{assmp}{Assumption}[section]
\newtheorem{remark}{Remark}[section]

\numberwithin{equation}{section} 

\renewcommand{\geq}{\geqslant}
\renewcommand{\leq}{\leqslant}

\newcommand{\E}{{\mathbb E}}

\newcommand{\R}{\ensuremath{\operatorname{\mathbb{R}}}}

\newcommand{\dd}{\ensuremath{\operatorname{d}\! }}
\newcommand{\dt}{\ensuremath{\operatorname{d}\! t}}
\newcommand{\de}{\ensuremath{\operatorname{d}\! e}}
\newcommand{\ds}{\ensuremath{\operatorname{d}\! s}}

\newcommand{\dw}{\ensuremath{\operatorname{d}\! W}}
\newcommand{\dn}{\ensuremath{\operatorname{d}\! N}}

\newcommand{\nn}{\nonumber}

\newcommand{\cM}{\ensuremath{\mathcal{M}}}

\newcommand{\ol}{\mathcal{L}}

\newcommand{\ou}{\mathcal{U}}

\newcommand{\cE}{\ensuremath{\mathcal{E}}}

\newcommand{\al}{\alpha_{t-}}

\usepackage{comment}
\excludecomment{extra}

\begin{document}

\title 
{Optimal mean-variance portfolio selection under regime-switching-induced stock price shocks}
\author{Xiaomin Shi\thanks{School of Statistics and Mathematics, Shandong University of Finance and Economics, Jinan
250100, China. Email: \url{shixm@mail.sdu.edu.cn}}
\and Zuo Quan Xu\thanks{Department of Applied Mathematics, The Hong Kong Polytechnic University, Kowloon, Hong Kong, China. Email: \url{maxu@polyu.edu.hk}}}
\maketitle
\begin{abstract}
In this paper, we investigate mean-variance (MV) portfolio selection problems with jumps in a regime-switching financial model. The novelty of our approach lies in allowing not only the market parameters --- such as the interest rate, appreciation rate, volatility, and jump intensity --- to depend on the market regime, but also in permitting stock prices to experience jumps when the market regime switches, in addition to the usual micro-level jumps. This modeling choice is motivated by empirical observations that stock prices often exhibit sharp declines when the market shifts from a ``bullish'' to a ``bearish'' regime, and vice versa.
By employing the completion-of-squares technique, we derive the optimal portfolio strategy and the efficient frontier, both of which are characterized by three systems of multi-dimensional ordinary differential equations (ODEs). Among these, two systems are linear, while the first one is an $\ell$-dimensional, fully coupled, and highly nonlinear Riccati equation. In the absence of regime-switching-induced stock price shocks, these systems reduce to simple linear ODEs. Thus, the introduction of regime-switching-induced stock price shocks adds significant complexity and challenges to our model.
Additionally, we explore the MV problem under a no-shorting constraint. In this case, the corresponding Riccati equation becomes a $2\ell$-dimensional, fully coupled, nonlinear ODE, for which we establish solvability. The solution is then used to explicitly express the optimal portfolio and the efficient frontier.
\end{abstract}

\addcontentsline{toc}{section}{\hspace*{1.8em}Abstract}
\section{Introduction}
Research on the mean-variance (MV) portfolio selection problems dates back to the Nobel-Prize-winning work Markowitz \cite{Ma} in the single-period setting. By embedding the problem into a tractable auxiliary stochastic linear-quadratic optimal control problem, Li and Ng \cite{LN} and Zhou and Li \cite{ZL} made a breakthrough on MV problems, respectively, in multi-period and continuous-time settings. Since then, abundant researches on dynamic MV problem have been conducted in more complicated and incomplete financial markets. For instance, Bielecki, Jin, Pliska and Zhou \cite{BJPZ} solved the MV problem with bankruptcy prohibition, Hou and Xu \cite{HX16} incorporated intractable claims, Lim \cite{Lim} allowed for jumps in the underlying assets, {Liu et al \cite{LYLST}} and Yu \cite{Yu} worked in a random horizon market, Zhou and Yin \cite{ZY} featured assets in a regime-switching market.

Li, Zhou and Lim \cite{LZL} solved the MV problem with no-shorting constraints when the wealth dynamic is continuous. They studied the corresponding Hamilton-Jacobi-Bellman (HJB) equation and gave an explicit viscosity solution of it in terms of two decoupled Riccati equations.
The optimal wealth process was implicitly shown to stay in the same region all the time as the initial wealth level. This phenomenon breaks down when it comes to MV problem in a jump diffusion model with no-shorting constraints, and two fully coupled Riccati equations have to be used to delineate the corresponding HJB equations and then the optimal portfolio, see, e.g., Shi and Xu \cite{SX}.

In this paper, we generalize our previous model \cite{SX} to a regime-switching market, where a continuous time Markov chain is employed to reflect the market trends, such as ``bullish market'' and ``bearish market'', or more intermediate regimes between the two extremes.

The study of MV problem with regime-switching was initiated by Zhou and Yin \cite{ZY}, followed notably by Chen and Yao \cite{CY} for short-selling prohibition market, Chen, Yao and Yin \cite{CYY}, Xie \cite{Xie} for asset liability management problem, and Wu \cite{Wu} for a jump diffusion model, etc. The optimal portfolio and efficient frontier were derived explicitly in these papers based on some multi-dimensional Riccati equations whose solvability is evident since they are all linear ordinary differential equations (ODEs). 
{Elliott, Siu, and Badescu~\cite{ESB} studied the MV portfolio selection problem in a hidden Markov regime-switching model, where the state of the underlying Markov chain is unobservable to the investor. By exploiting the separation principle, they addressed the filtering and estimation problem for the Markov chain, and subsequently determined the optimal portfolio using the stochastic maximum principle. Notably, these works focus on pre-commitment strategies, which will also be adopted in this paper. 
There is also a substantial body of literature on time-consistent MV portfolio selection in regime-switching models; see, e.g., Chen, Chiu and Wong \cite{CCW}, Wang, Jin and Wei \cite{WJW19} and references therein.}

In addition to the regime-switching jump diffusion modulated stock price for each fixed regime (market mode) and the no-shorting constraint, the most salient feature of our model that distinguishes \cite{ZY} is as follows: The stock price posses a corresponding jump when the Markov chain $\alpha$ changes its regime. This model has been applied in Busch, Korn and Seifried \cite{BKS} and Gassiat, Gozzi and Pham \cite{GGP} where optimal investment and consumption problems in a large investor model and an illiquid market were studied respectively. In a regime-switching model, the market parameters, such as the interest rate, the appreciation, volatility and jump intensity rates are modulated by a continuous time Markov chain. And there are indeed two state processes, the wealth dynamic $X$ and $\alpha$ in such models. Incorporating a possible jump when the Markov chain $\alpha$ changes its regime in the stock price, consequently $X$, could help us to capture more information of $\alpha$ via the {optional quadratic variation} $[X,\alpha]$ after noting the representation \eqref{alpha}. Needless to say, this will be more closer to the real financial market as well. Mathematically, this will render the corresponding multi-dimensional Riccati equations no longer linear, and more involved to study; see \eqref{P} and \eqref{Riccati}. Therefore, the introduction of regime-switching-induced stock price shocks introduces significant complexities and challenges to our model.
To the best of our knowledge, the nonlinear Riccati equations \eqref{P} and \eqref{Riccati} are new in the literature.
Unfortunately, due to technical limitations, we are compelled to treat all coefficients as deterministic functions of regime, time, and jumps. This simplification transforms the Riccati equations into ODEs. Addressing the case of general stochastic coefficients, where the Riccati equations will be backward stochastic differential equations (BSDEs), presents a significantly more complex challenge.

The remainder of this paper is organized as follows. In Section \ref{sec:pf}, we formulate a MV
portfolio problem in a regime-switching market under regime-switching-induced stock price shocks; and we then resolves its feasibility issue in this part. In Section \ref{sec:noconstraint} and Section \ref{sec:noshort}, we obtain the optimal strategies and efficient frontiers, which are characterized by systems of multi-dimensional, fully coupled, and highly nonlinear ODEs, for the MV problems without trading constraint and with short-selling prohibition respectively.
Section \ref{sec:conclusion} concludes the paper.

\section{Problem formulation and its feasibility issue}\label{sec:pf}
Throughout the paper, let $T$ be a fixed positive constant, standing for the investment horizon.
We denote by $\R^m$ the set of all $m$-dimensional column vectors, by $\R^m_+$ the set of vectors in $\R^m$ whose components are nonnegative, by $\R^m_{++}$ the set of vectors in $\R^m$ whose components are positive, by $\R^{m\times n}$ the set of $m\times n$ real matrices, and by $\mathbb{S}^n$ the set of symmetric $n\times n$ real matrices. For any vector $Y$, we denote $Y_i$ as its $i$-th component. For any vector or matrix $M=(M_{ij})$, we denote its transpose by $M^{\top}$, and its norm by $|M|=\sqrt{\sum_{ij}M_{ij}^2}$. If $M\in\mathbb{S}^n$ is positive definite (resp. positive semidefinite), we write $M>$ (resp. $\geq$) $0.$ We write $A>$ (resp. $\geq$) $B$ if $A, B\in\mathbb{S}^n$ and $A-B>$ (resp. $\geq$) $0.$ We use the standard notations $x^+=\max\{x, 0\}$ and $x^-=\max\{-x, 0\}$ for $x\in\R$.

Let $(\Omega, \mathcal F, \mathbb{F}, \mathbb{P})$ be a fixed complete filtered probability space, on which is defined three independent random processes: a standard $n_1$-dimensional Brownian motion $W_t=(W_{1,t}, \ldots, W_{n_1,t})^{\top}$, an $n_2$-dimensional Poisson random measure $\Phi(\dt,\de)$ defined on $\R_+\times\:\cE$ with a stationary compensator (intensity measure) given by $\nu(\de)\dt=(\nu_1(\de),\ldots,\nu_{n_2}(\de))^{\top}\dt$ satisfying $\sum_{j=1}^{n_2}\nu_j(\cE)<\infty$, where $\mathcal{\cE}\subseteq \R^{n_3}\setminus\{0\}$ is a nonempty Borel subset of $\R^{n_3}$ and $\mathcal{B}(\cE)$ denotes the Borel $\sigma$-algebra generated by $\cE$, and a {time-homogeneous continuous-time Markov chain} $\alpha$ valued in a finite regime space $\cM:=\{1,2,\ldots,\ell\}$ with $\ell\geq 1$. The Markov chain has a generator $Q=(q^{ij})_{\ell\times\ell}$ with $q^{ij}\geq 0$ for $i\neq j$ and $\sum_{j=1}^{\ell}q^{ij}=0$ for every $i\in\cM$. For $i\neq j\in\cM$, we will associate to the Markov chain $\alpha$ a Poisson process $N^{ij}$ with intensity rate $q^{ij}\geq 0$, such that a switch from regime $i$ to $j$ corresponds to a jump of $N^{ij}$ when $\alpha$ is in regime $i$. We denote by $\tilde \Phi(\dt,\de)$ the compensated Poisson random measure of $\Phi(\dt,\de)$, by $\tilde N^{ij}$ the compensated Poisson process of $N^{ij}$.
According to \cite[Eq. (13.71)]{Br}, the Markov chain $\alpha$ admits the representation
\begin{align}\label{alpha}
\dd\alpha_t &=\sum_{j=1}^{\ell}(j-\alpha_{t-})\dn_t^{\alpha_{t-}j} =\sum_{i=1}^{\ell}\sum_{j=1}^{\ell}(j-i)\mathbf{1}_{\{\alpha_{t-}=i\}}\dn_t^{ij}.
\end{align}
The Markov chain $\alpha$ is served to model the regime of the underlying financial market. {The filtration $\mathbb{F}$ is generated by all the processes defined above, including the Brownian motion, the Poisson random measure, the Markov chain and the Poisson processes $\{N^{ij}\}_{i,j\in\cM}$, and is augmented with all the $\mathbb{P}$-null sets.}
All the equations and inequalities in subsequent analysis shall be understood in the sense that $\dd\mathbb{P}$-a.s. or $\dd\nu$-a.e. or $\dt$-a.e. or $\dt\otimes \dd\nu$-a.e. etc, which shall be seen from their expressions.

\subsection{Problem formulation}
Suppose that there are one money account whose price $S_{0,t}$ is given by
\begin{align*}
\dd S_{0,t}=r_t^{\al}S_{0,t}\dt,
\end{align*}
and $m$ tradable stocks in the financial market.
In the market regime $\alpha_t=i$, the $k$-th ($k=1,\ldots,m$) stock price follows the jump-diffusion dynamics
\begin{equation*}
\dd S_{k,t}=S_{k,t}\Big(\mu_{k,t}^{i}\dt+\sum_{j=1}^{n_1}\sigma_{kj,t}^{i}\dw_{j,t}
+\sum_{j=1}^{n_2}\int_{\cE}\beta^i_{kj,t}\tilde \Phi_{j}(\dt,\de)\Big).
\end{equation*}
Moreover, at the time $t$ of transition from $\al=i$ to $\alpha_t=j$, the $k$-th stock price changes as follows:
\begin{equation*}
S_{k,t}=S_{k,t-}(1+\gamma_{k,t}^{i,j}).
\end{equation*}
Overall, the $k$-th stock price is governed by
a regime-switching-induced jump model:
\begin{align*}
\dd S_{k,t}&=S_{k,t}\Big(\mu_{k,t}^{\alpha_{t-}}\dt+\sum_{j=1}^{n_1}\sigma_{kj,t}^{\alpha_{t-}}\dw_{j,t} +\sum_{j=1}^{n_2}\int_{\cE}\beta^{\al}_{kj,t}\tilde \Phi_{j}(\dt,\de)
+\gamma_{k,t}^{\alpha_{t-},\alpha_{t}}\dn^{\alpha_{t-},\alpha_t}_t\Big),
\end{align*}
here and hereafter we use the notation 
$\gamma_{k,t}^{\alpha_{t-},\alpha_{t}}\dn^{\alpha_{t-},\alpha_t}_t$ to stand for
$\sum_{j\neq\alpha_{t-}}\gamma_{k,t}^{\alpha_{t-},j}\dn^{\alpha_{t-},j}_t$.

For each $i,j\in\cM$,
denote by $\mu^i_t:=(\mu^i_{1,t},\ldots,\mu^i_{m,t})^{\top}$, $\sigma^i_{t}:=({\sigma^i_{kj,t}})_{m\times n_1}$, $\gamma^{i,j}_t:=(\gamma^{i,j}_{1,t},\ldots,\gamma^{i,j}_{m,t})^{\top}$, $\beta^i_{t}(e):=(\beta^i_{kk',t}(e))_{m\times n_2}$, and $\beta^i_{k',t}$ the $k'$-th column of $\beta^i_{t}$, $k'=1,\ldots,n_2$.
Throughout the paper, we put the following assumption on the market parameters without claim.

\begin{assmp}\label{assu1}
For each $i,j=1,\ldots,\ell$, $k=1,\ldots,m$, $k'=1,\ldots,n_2$, the coefficients
$r^i_t$, $\mu_{t}^i$, $\sigma_{t}^i$, $\gamma_{t}^{i,j}$ (resp. $\beta^i_{t}$) are deterministic, Borel-measurable and bounded functions on $[0,T]$ (resp. $[0,T]\times\cE$) valued in $\R$, $\R^m$, $\R^{m\times n_1}$, $\R^m$ respectively (resp. valued in $m\times n_2$), with $\beta^i_{kk',t}(e)>-1$ and $\gamma_{k,t}^{i,j}>-1$. And there exists a constant $\delta>0$ such that
\begin{align*}
\Sigma^i_t &:=\sigma^{i}_t(\sigma^{i}_t)^{\top}
+\sum_{k'=1}^{n_2}\int_{\cE}\beta^i_{k',t}(e)\beta^i_{k',t}(e)^{\top}\nu_{k'}(\de) +\sum\limits_{j\neq i}^{\ell}q^{ij}\gamma^{i,j}_t(\gamma^{i,j}_t)^{\top}
\geq\delta\mathbf{1}_m, \ \forall t\in[0,T], \ i\in\cM,
\end{align*}
where $\mathbf{1}_m$ denotes the $m\times m$ identity matrix.
\end{assmp}

We consider that a small investor, whose actions cannot affect the trading asset prices. {He can observe all the driving stochastic processes, especially the Markov chain.} He will decide at every time $t\in[0, T]$ the amount $\pi_{k,t}$ of his wealth to invest in the $k$-th risky asset, $k=1, \ldots, m$. Any {$\mathbb F$-predictable} vector process $\pi:=(\pi_1, \ldots, \pi_m)^{\top}$ is called a portfolio of the investor. Then the investor's self-financing wealth process $X$ corresponding to a portfolio $\pi$ is the strong solution to the SDE:
\begin{align} 
\label{wealth}
\begin{cases}
\dd X_t=[r_t^{\alpha_{t-}}X_{t-}+\pi_t^{\top}\mu_t^{\alpha_{t-}}] \dt+\pi_t^{\top}\sigma_t^{\alpha_{t-}}\dw_t\medskip\\
\qquad\quad+\int_{\cE}\pi_t^{\top}\beta^{\al}_t(e)\tilde \Phi(dt,\de)
+\pi_t^{\top}\gamma_t^{\alpha_{t-},\alpha_{t}}\dn_t^{\alpha_{t-},\alpha_{t}}, \medskip\\
X_0=x, \ \alpha_0=i_0.
\end{cases}
\end{align}
{Zhang, Elliott and Siu \cite{ZES} studied a stochastic control problem where the regime-switching jump shocks were also considered in the state process. In their model, for each $j\in\cM$, the coefficients before $N^{i,j}, \ i\in\cM$ were the same as they regarded $\sum_{i\in\cM,i\neq j} N^{i,j}$ as an integral whole. Whereas in our model, we regard $N^{i,j}, \ i,j\in\cM, \ i\neq j$ as different jump processes, so the coefficients before them could be totally different. From this perspective, our regime-switching-induced jump model is more refined.}
The admissible portfolio set is defined as
\begin{align*}
\mathcal U=\Big\{\pi:[0,T]\times\Omega\to\Pi\;\Big|\; &\pi \ \mbox{is predictable}, \int_0^T\E[\pi_t^2]\dt<\infty \Big\},
\end{align*}
where $\Pi=\R^{m}$ or $\Pi=\R_{+}^{m}$ denotes the trading constraint set. Financially speaking, $\Pi=\R^{m}$ means that there are no trading constraints, while $\Pi=\R_{+}^{m}$ means that shorting-selling of stocks is prohibited.\footnote{Extending our model to the case where only a subset of stocks are permitted for short selling presents no significant theoretical challenges.}
For any $\pi\in \mathcal{U}$, the SDE \eqref{wealth} has a unique strong solution $X$. Denote by $X^{\pi}$ the wealth process \eqref{wealth} whenever it is necessary to indicate its dependence of the portfolio $\pi$.

\par
For a given expectation level $z\in\R$, the investor's MV problem is to
\begin{align}
\mathrm{Minimize}&\quad \mathrm{Var}(X_T^{\pi})=\E[(X_T^{\pi})^{2}-z^2]%
, \nn\medskip\\
\mathrm{s.t.} &\quad
\begin{cases}
\E[X_T^{\pi}]=z, \smallskip\\
\pi\in \mathcal{U}.
\end{cases}
\label{optm}%
\end{align}
Denote by $$\Pi_{z}=\{\pi\;|\;\pi\in\ou, \ \mbox{and} \ \E [X_T^{\pi}]=z\}.$$ 
The MV problem \eqref{optm} is called feasible if
$\Pi_{z}$ is not empty. Correspondingly, any $\pi\in\Pi_{z}$ is called a feasible portfolio for the MV problem \eqref{optm}.

An optimal strategy $\pi^{\ast}$ to \eqref{optm}
is called an efficient strategy corresponding to the level $z$. And $\Big(\sqrt{\mathrm{Var}(X_{T}^{\pi^\ast})}, z\Big)$ is called an efficient point of the MV problem \eqref{optm}, whereas the set of all the efficient points is called the efficient frontier of the MV problem \eqref{optm}. This paper will derive the explicit optimal portfolio for each $z\in\R$ and explicit efficient frontier of the MV problem \eqref{optm}.

\subsection{Feasibility}\label{sec:feas}
In this section, we address the feasibility issue of the MV problem \eqref{optm}, that is, when the set $\Pi_{z}$ is not empty.

The following generalized It\^o's formula for jump diffusion process (see, e.g., \cite{OS}) is crucial for our subsequent analysis.
\begin{lemma}[It\^o's formula]
Given a state process
\begin{align*} 
\dd X_t =&~ b^{\al}_t(X_t)\dt+\sigma^{\al}_t(X_t)\dw_t +\int_{\cE}\xi^{\al}_{t}(X_{t-},e) \Phi(\dt,\de)
+\rho_t^{\al,\alpha_t}(X_{t-})\dd N^{\al,\alpha_t}_t,
\end{align*}
and $\ell$ functions $\varphi^i(\cdot,\cdot)\in C^2([0,T]\times\R), \ i\in\cM$, we have
\begin{align*} 
\dd \varphi^{\alpha_t}(t,X_t)=&~\ol\varphi^{\al}(t,X_t)\dt+\varphi^{\al}_x(t,X_t)\sigma^{\al}_t(X_t)\dw_t\medskip\\
&~+\sum_{k=1}^{n_2}\int_{\cE}[\varphi^{\al}(t,X_{t-}+\xi^{\al}_{k,t}(X_{t-},e))
-\varphi^{\al}(t,X_{t-})] \tilde \Phi_k(\dt,\de)\qquad\qquad\medskip\\
&~+\sum_{j\neq\al}^{\ell}[\varphi^j(t,X_{t-}+\rho^{\al,j}_t(X_{t-}))-\varphi^{\al}(t,X_{t-})]\dd \tilde N^{\al,j}_t,
\end{align*}
where
\begin{align*}
\ol\varphi^i(t,x):=&~\frac{1}{2}\varphi^i_{xx}(t,x)|\sigma^i_t|^2
+\varphi^i_x(t,x)b^i_t(x)\medskip\\
&~+\sum_{k=1}^{n_2}\int_{\cE}[\varphi^i(t,x+\xi^i_{k,t}(x,e))
-\varphi^i(t,x)]\nu_k(\de)\medskip\\
&~+\sum_{j\neq i}^{\ell}q^{ij}[\varphi^j(t,x+\rho_{t}^{ij}(x))
-\varphi^i(t,x)].
\end{align*}
\end{lemma}
The following result resolves the feasibility issue of the MV problem \eqref{optm}.
\begin{lemma}
Let $\psi=(\psi^1,\ldots,\psi^\ell)$ be the unique solution to the following system of linear ODEs:
\begin{align}
\label{psi}
\begin{cases}
\dot{\psi^{i}_t}=-r^{i}_t\psi^{i}_t-\sum\limits_{j\neq i}^{\ell}q^{ij}(\psi^j_t-\psi^i_t)\medskip\\
\psi^{i}_T=1, \ i\in\cM.
\end{cases}
\end{align}
Then the MV problem \eqref{optm} is feasible for every $z\in\R$ if and only if
\begin{equation}
\sum_{k=1}^{m}\E\bigg[\int_0^T\Big|\psi^{\al}_t\mu^{\al}_{k,t}+\sum_{j\neq\al}q^{\al j}\psi^j_t\gamma^{\al,j}_{k,t}\Big|\dt\bigg]>0 \label{assufeas}
\end{equation}
when $\Pi=\R^m$, or
\begin{equation}
\sum_{k=1}^{m}\E\bigg[\int_0^T\Big(\psi^{\al}_t\mu^{\al}_{k,t}+\sum_{j\neq\al}q^{\al j}\psi^j_{t}\gamma^{\al,j}_{k,t}\Big)^+\dt\bigg]>0 \label{assufeas2}
\end{equation}
when $\Pi=\R^m_+$.
\end{lemma}
The proof is similar to \cite[Theorem 5.3]{HSX}, so we omit it.

The way to solve the MV problem \eqref{optm} is rather clear nowadays. To deal with the budget constraint $\E[X_T]=z$, we introduce a Lagrange multiplier $-2\lambda\in\R$ and introducing the following standard stochastic linear quadratic (LQ) control problem:
\begin{equation}\label{optmun}
V(x,i_0;\lambda):=\inf_{\pi\in\mathcal{U}}{\mathbb{E}}[(X_T-\lambda)^{2}]
\end{equation}
for $(x,i_0,\lambda)\in \R\times\cM\times\R$.
According to the Lagrange duality theorem (see Luenberger \cite{Lu}), the problem \eqref{optm} is linked to the problem \eqref{optmun} by
\begin{equation*}
\inf_{\pi\in\Pi_{z}}\mathrm{Var}(X_T)=\sup_{\lambda\in\R}[V(x,i_0;\lambda)-(\lambda-z)^{2} ].
\end{equation*}
So we can solve the MV problem \eqref{optm} by a two-step procedure: Firstly determine the function $V(x,i_0;\lambda)$, and then try to find a $\lambda^{*}\in\R$ to maximize the function $\lambda\mapsto V(x,i_0;\lambda)-(\lambda-z)^{2}$ over $\R$.

\section{MV portfolio selection without trading constraint}\label{sec:noconstraint} 
In this section, we study the MV portfolio selection problem \eqref{optm} without trading constraint, i.e. $\Pi=\R^m$.

First for any $P=(P^1,\ldots,P^{\ell})^{\top}\in\R^{\ell}_{++}$ and $h=(h^1,\ldots,h^{\ell})^{\top}\in\R^{\ell}$, we define the following mappings:
\begin{align*} 
\mathscr{M}^i_t(P):=&~P^i\mu_t^i+\sum_{j\neq i}^{\ell}q^{ij}P^j\gamma^{i,j}_t,\medskip\\
\mathscr{N}^i_t(P):=&~P^i\sigma^i_t(\sigma^i_t)^{\top}
+P^i\sum_{k=1}^{n_2}\int_{\cE}\beta^i_{k,t}(e)\beta^i_{k,t}(e)^{\top}\nu_{j}(\de) +\sum\limits_{j\neq i}^{\ell}q^{ij}P^j\gamma^{i,j}_t(\gamma^{i,j}_t)^{\top},\medskip\\
\mathscr{R}^i_t(P,h):= &~\sum_{j\neq i}^{\ell}\Big(q^{ij}P^j\gamma^{i,j}_t(h^j-h^i)\Big).
\end{align*}
Then consider the following system of Riccati equations, which consists of an $\ell$-dimensional fully coupled ODE:
\begin{align} 
\label{P}
\begin{cases}
\dot{P^i_t}
=-\Big[2r^i_tP^i_t-\mathscr{M}_t^i(P_t)^{\top}\mathscr{N}_t^i(P_t)^{-1}\mathscr{M}_t^i(P_t) +\sum\limits_{j\neq i}^{\ell}q^{ij}(P^j_t-P^i_t)\Big]\medskip\\
P_T^i=1, \medskip\\
P^i_t>0, \ t\in[0,T], \ i\in\cM.
\end{cases}
\end{align}
and two $\ell$-dimensional linear ODEs:
\begin{align} 
\label{h}
\begin{cases}
\dot{h^{i}_t}=r^{i}_th^{i}_t+\frac{1}{P^{i}_t}\mathscr{M}_t^i(P_t)^{\top}\mathscr{N}_t^i(P_t)^{-1}\mathscr{R}^i_t(P_t,h_t) -\frac{1}{P^i_t}\sum\limits_{j\neq i}^{\ell}q^{ij}P^j_t(h^j_t-h^i_t)\medskip\\
h^{i}_T=1, \ i\in\cM,
\end{cases}
\end{align}
and
\begin{equation}
\label{K}
\begin{cases}
\dot{K^{i}_t}=\mathscr{R}^i_t(P_t,h_t)^{\top}\mathscr{N}_t^i(P_t)^{-1}\mathscr{R}^i_t(P_t,h_t) -\sum\limits_{j\neq i}^{\ell}\Big(q^{ij}P^j_t(h^j_t-h^i_t)^2\Big)-\sum\limits_{j\neq i}^{\ell}q^{ij}(K^j_t-K^i_t)\medskip\\
K^{i}_T=0, \ i\in\cM.
\end{cases}
\end{equation}

\begin{remark}
Note that \eqref{P} is a highly nonlinear ODE. By contrast, in the absence of regime-switching-induced stock price shocks, i.e. $\gamma^{i,j}=0$, $ i,j\in\cM$, \eqref{P} degenerates to a linear ODE and \eqref{P} and \eqref{h} degenerate to \cite[Eq. (4.2) and (4.3)]{ZY}. Therefore, the introduction of regime-switching-induced stock price shocks introduces significant complexities and challenges to our model.
\end{remark}

The ODEs \eqref{P}, \eqref{h}, \eqref{K} are coupled together, but we can solve them one by one, since the former does not depend on the later. In particular the latter two are linear ones, so their solvability is immediately ready.

\begin{thm}\label{Th:P}
Under Assumption \ref{assu1}, there exist, respectively, unique solutions $P=(P^1,\ldots,P^{\ell})$ to the Riccati equation \eqref{P}, $h=(h^1,\ldots,h^{\ell})$ to the ODE \eqref{h} and $K=(K^1,\ldots,K^{\ell})$ to the ODE \eqref{K}, such that $P^i_t>0,\ t\in[0,T], \ i\in\cM$.
\end{thm}
The proof of Theorem \ref{Th:P} will be covered by Theorem \ref{existence} (see Remark \ref{special}), which handles a more general $2\ell$-dimensional Riccati equation, so we omit it.

\begin{remark}
If the interest rate $r_t$ is independent of the Markov chain $\alpha$, then the unique solutions to the ODEs \eqref{h} and \eqref{K} are given by
\[
h^i_t=e^{-\int_t^Tr_s\ds}, \ K^i_t=0, \ t\in[0,T], \ i\in\cM.
\]
\end{remark}

\begin{thm}\label{Th3}
Suppose Assumption \ref{assu1} {holds}.
Let $(P,h,K)$ be the unique solution to the ODE system \eqref{P}, \eqref{h}, \eqref{K}.
Then the LQ problem \eqref{optmun} admits a unique optimal control, as a feedback function of the time $t$, the state $X$, and the Markov chain $\alpha$, given by
\begin{align}\label{pistar1}
\pi^*(t,X_{t-},\alpha_{t-})=-\Big[\mathscr{N}^{\al}_t(P_t)\Big]^{-1} \Big[\mathscr{M}_t^{\al}(P_t)(X_{t-}-\lambda h_t^{\al})
-\lambda\mathscr{R}_t^{\al}(P_t,h_t)\Big],
\end{align}
and the corresponding optimal value is
\begin{align}\label{value1}
V(x,i_0;\lambda) 
&=P_0^{i_0}(x-\lambda h_0^{i_0})^2+\lambda^2K^{i_0}_0.
\end{align}
\end{thm}
\begin{proof}
For any $\pi\in\mathcal{U}$, 
applying It\^o's formula to
$P_t^{\alpha_{t}}(X_t-\lambda h_t^{\alpha_{t}})^2$ and $\lambda^2K^{\alpha_{t}}_t$ respectively, we get
\begin{align*}
\dd \Big[P_t^{\alpha_{t}}(X_t-\lambda h_t^{\alpha_{t}})^2\Big] 
&= \Big[(\pi_t-\pi^*_t)^{\top} \mathscr{N}^{\al}_t(P_t)(\pi_t-\pi^*_t)\medskip\\
&\qquad-\lambda ^2\mathscr{R}^{\al}_t(P_t,h_t)^{\top}\mathscr{N}_t^{\al}(P_t)^{-1}\mathscr{R}^{\al}_t(P_t,h_t)\medskip\\
&\qquad+\lambda ^2\sum_{j\neq\al}^{\ell}q^{\al j}P^j_t(h^j_t-h^{\al}_t)^2\Big]\medskip\\
&\quad\;+(\cdots)\dw_t+(\cdots) \tilde \Phi_k(\dt,\de)+(\cdots)\dd \tilde N^{\al,j},
\end{align*}
where $\pi^*$ is defined in \eqref{pistar1}, and
\begin{align*}
\dd \ [\lambda^2K^{\alpha_{t}}_t]&=\lambda ^2 \Big[\mathscr{R}^{\al}_t(P_t,h_t)^{\top}\mathscr{N}_t^{\al}(P_t)^{-1}\mathscr{R}^{\al}_t(P_t,h_t) -\sum_{j\neq\al}^{\ell}q^{\al j}P^j_t(h^j_t-h^{\al}_t)^2\Big]\dt\medskip\\
&\quad+(\cdots)\dd \tilde \Phi^{\al,j}.
\end{align*}
Summing the above two equations, integrating from $0$ to $T$ and taking expectation, we obtain, noting the terminal conditions of $P,h,K$ in \eqref{P}, \eqref{h}, \eqref{K},
\begin{align*} 
\E[(X_T-\lambda)^2] 
&=\E\Big[P_T^{\alpha_{T}}(X_T-\lambda h_T^{\alpha_{T}})^2+\lambda^2K^{\alpha_{T}}_T\Big]\medskip\\
&=P_0^{i_0}(x-\lambda h_0^{i_0})^2+\lambda^2K_0^{i_0} +\E\int_0^T (\pi_t-\pi^*_t)^{\top} \mathscr{N}^{\al}_t(P_t)(\pi_t-\pi^*_t)\dt.
\end{align*}
Since $\mathscr{N}^{\al}_t(P_t)>0$, it follows immediately that the optimal control is given by \eqref{pistar1} and the optimal value is given by \eqref{value1}, provided that the corresponding wealth equation \eqref{wealth} under the feedback control \eqref{pistar1} admits a solution. However, under \eqref{pistar1}, the equation \eqref{wealth} is a linear SDE with bounded coefficients, the existence of a unique strong solution is straightforward.
\end{proof}

\begin{remark}
By definition, the optimal value of the LQ problem \eqref{optmun} is nonnegative. But it is vague from the right-hand side of \eqref{value1} unless we can show $K_0^{i_0}\geq 0$.

Indeed, we do have $K_0^{i_0}\geq 0$. To see this, we first get from \eqref{K} that
\begin{align*} 
K_0^{i_0}&=-\E\int_{0}^T\Big[\mathscr{R}^{\al}_t(P_t,h_t)^{\top}\mathscr{N}_t^{\al}(P_t)^{-1}
\mathscr{R}^{\al}_t(P_t,h_t) -\sum\limits_{j\neq \al}^{\ell}\Big(q^{\al j}P^j_t(h^j_t-h^{\al}_t)^2\Big)\Big]\dt.
\end{align*}
Hence, it suffices to show
\begin{equation*}
\mathscr{R}^i_t(P_t,h_t)^{\top}\mathscr{N}_t^i(P_t)^{-1}\mathscr{R}^i_t(P_t,h_t)
\leq\xi:=\sum\limits_{j\neq i}^{\ell}\Big(q^{ij}P^j_t(h^j_t-h^i_t)^2\Big).
\end{equation*}
Clearly $\xi\geq0$.
According to \cite[Lemma A.1]{HSX2}, we know (where the argument $(P_t,h_t)$ is suppressed)
\begin{equation*}
\Big(\xi
-(\mathscr{R}^i)^{\top}(\mathscr{N}^i)^{-1}\mathscr{R}^i\Big)\mathrm{det}(\mathscr{N}^i)= \xi
\mathrm{det}\Big(\mathscr{N}^i-\frac{1}{\xi}\mathscr{R}^i(\mathscr{R}^i)^{\top}\Big).
\end{equation*}
From the definition of $\mathscr{N}^i$, we immediately get
\begin{equation*}
\mathscr{N}^i-\frac{1}{\xi}\mathscr{R}^i(\mathscr{R}^i)^{\top}\geq \sum\limits_{j\neq i}^{\ell}q^{ij}P^j\gamma^{i,j}_t(\gamma^{i,j}_t)^{\top}
-\frac{1}{\xi}\mathscr{R}^i(\mathscr{R}^i)^{\top}.
\end{equation*}
So it is sufficient to show that
\begin{equation*}
\xi\sum\limits_{j\neq i}^{\ell}q^{ij}P^j\gamma^{i,j}_t(\gamma^{i,j}_t)^{\top}
-\mathscr{R}^i(\mathscr{R}^i)^{\top}\geq 0.
\end{equation*}
In fact, for any $y\in\R^m$, it follows from the Cauchy-Schwarz inequality that
\begin{align*} 
y^{\top}\mathscr{R}^i(\mathscr{R}^i)^{\top}y 
&=\Big|\sum_{j\neq i}^{\ell}\Big(q^{ij}P^j(h^j-h^i)y^{\top}\gamma^{i,j}_t\Big)\Big|^2\medskip\\
&\leq \Big[\sum_{j\neq i}^{\ell}\Big(q^{ij}P^j(h^j-h^i)^2\Big)\Big]\Big[\sum_{j\neq i}^{\ell}\Big(q^{ij}P^j|y^{\top}\gamma^{ij}|^2\Big)\Big]\medskip\\
&=\xi y^{\top}\sum\limits_{j\neq i}^{\ell}q^{ij}P^j\gamma^{i,j}_t(\gamma^{i,j}_t)^{\top}y.
\end{align*}
This completes the proof of $K_0^{i_0}\geq 0$.
\end{remark}

Next we proceed to derive the optimal portfolio and efficient frontier for the original MV problem \eqref{optm}. {Recall that $z$ is the given mean level, $x$ is the initial wealth, and $i_0$ is the initial state of the Markov chain.}

\begin{thm}\label{Th:efficient}
Under the conditions of Theorem \ref{Th3} and suppose the feasible condition \eqref{assufeas} holds. Then we have
\begin{align}\label{nonnegative}
1-P^{i_0}_0(h^{i_0}_0)^2-K^{i_0}_0>0.
\end{align}
Moreover, the feedback portfolio defined by
\begin{align}\label{efficient}
\pi^*(t,X_{t-},\alpha_{t-})=-\Big[\mathscr{N}^{\al}_t(P_t)\Big]^{-1} \Big[\mathscr{M}_t^{\al}(P_t)(X_{t-}-\lambda^*h_t^{\al})
-\lambda^*\mathscr{R}_t^{\al}(P_t,h_t)\Big],
\end{align}
is optimal for the MV problem \eqref{optm}, where
\begin{align}\label{lambdastar}
\lambda^*=\frac{z-xP^{i_0}_0h^{i_0}_0}{1-P^{i_0}_0(h^{i_0}_0)^2-K^{i_0}_0}.
\end{align}
Furthermore, the efficient frontier to the problem \eqref{optm} is determined by
\begin{align}\label{efficientfron}
\mathrm{Var}(X^{\pi^{*}}_T)&=\frac{P^{i_0}_0(h^{i_0}_0)^2+K^{i_0}_0}{1-P^{i_0}_0(h^{i_0}_0)^2-K^{i_0}_0} 
\bigg(\E[X^{\pi^{*}}_T]-\frac{P^{i_0}_0h^{i_0}_0}{P^{i_0}_0(h^{i_0}_0)^2
+K^{i_0}_0}x\bigg)^2 +\frac{P^{i_0}_0K^{i_0}_0}{P^{i_0}_0(h^{i_0}_0)^2+K^{i_0}_0}x^2,
\end{align}
where
$$\E[X^{\pi^{*}}_T]\geq \frac{P^{i_0}_0h^{i_0}_0}{P^{i_0}_0(h^{i_0}_0)^2+K^{i_0}_0}x.$$
\end{thm}
\begin{proof}
To prove \eqref{nonnegative}, we
apply It\^o's formula to $P^{\alpha_t}_t(h^{\alpha_t}_t)^2+K^{\alpha_t}_t$, obtaining
\begin{align*} 
1-P^{i_0}_0(h^{i_0}_0)^2-K^{i_0}_0 
&=\E\int_0^T \Big(h^{\al}_t\mathscr{M}_t^{\al}(P_t)+\mathscr{R}_t^{\al}(P_t,h_t)\Big)^{\top}\medskip\\
&\qquad\qquad\;\times
\mathscr{N}_t^{\al}(P_t)^{-1}\Big(h^{\al}_t\mathscr{M}_t^{\al}(P_t)+\mathscr{R}_t^{\al}(P_t,h_t)\Big)\dt.
\end{align*}
So $1-P^{i_0}_0(h^{i_0}_0)^2-K^{i_0}_0\geq 0$, and the equality holds only if
\begin{align}\label{contradiction}
\E\int_0^T|h^{\al}_t\mathscr{M}_t^{\al}(P_t)+\mathscr{R}_t^{\al}(P_t,h_t)|\dt=0.
\end{align}
Suppose that \eqref{contradiction} holds.
Applying It\^o's formula to $P^{\alpha_t}_th^{\alpha_t}_t$,
we get
\begin{equation}\label{Ph}
\begin{cases}
\dd \ [P^{\alpha_t}_th^{\alpha_t}_t] 
=-r^{\al}_tP^{\al}_th^{\al}_t\dt -\sum_{j\neq\al}(P^{j}_th^{j}_t-P^{\al}_th^{\al}_t)\dd \tilde N^{\al,j}_t,\medskip\\
P^{\alpha_T}_Th^{\alpha_T}_T=1.
\end{cases}
\end{equation}
On the other hand, for the unique solution $\psi$ to \eqref{psi}, we also have
\begin{equation}\label{psial}
\begin{cases}
\dd \psi^{\alpha_t}_t
=-r^{\al}_t\psi^{\al}_t\dt-\sum_{j\neq\al}(\psi^{j}_t-\psi^{\al}_t)\dd \tilde N^{\al,j}_t,\medskip\\
\psi^{\alpha_T}_T=1.
\end{cases}
\end{equation}
By the uniqueness of solutions to the linear BSDEs \eqref{Ph} and \eqref{psial}, it must holds that $P^{\alpha_t}_th^{\alpha_t}_t=\psi^{\alpha_t}_t$ and $P^{j}_th^{j}_t=\psi^{j}_t, \ t\in[0,T], \ j\in\cM$. Consequently,
\begin{align*} 
h^{\al}_t\mathscr{M}_t^{\al}(P_t)+\mathscr{R}_t^{\al}(P_t,h_t) 
&=h^{\al}_tP^{\al}_t\mu^{\al}_t+\sum_{j\neq\al}q^{\al j}P^j_th^j_t\gamma^{\al,j}_t\medskip\\
&=\psi^{\al}_t\mu^{\al}_t+\sum_{j\neq\al}q^{\al j}\psi^j_t\gamma^{\al,j}_t,
\end{align*}
Then from \eqref{contradiction}, we immediately get
\begin{align*}
\E\int_0^T\Big|\psi^{\al}_t\mu^{\al}_t+\sum_{j\neq\al}q^{\al j}\psi^j_t\gamma^{\al,j}_t\Big|\dt=0,
\end{align*}
contradicting to \eqref{assufeas}. This completes the proof of \eqref{nonnegative}.

Using \eqref{value1} and \eqref{nonnegative},
a simple calculation shows
\begin{align*} 
\sup_{\lambda\in\R}[V(x,i_0;\lambda)-(\lambda -z)^{2}] 
&=\sup_{\lambda \in\R}[P_0^{i_0}(x-\lambda h_0^{i_0})^2+\lambda ^2K^{i_0}_0-(\lambda -z)^2]\medskip\\
&=V(x,i_0;\lambda^*)-(\lambda^*-z)^{2}\bigskip\\
&=\frac{P^{i_0}_0(h^{i_0}_0)^2+K^{i_0}_0}{1-P^{i_0}_0(h^{i_0}_0)^2-K^{i_0}_0}
\bigg(z-\frac{P^{i_0}_0h^{i_0}_0}{P^{i_0}_0(h^{i_0}_0)^2+K^{i_0}_0}x\bigg)^2\medskip\\
&\quad\;+\frac{P^{i_0}_0K^{i_0}_0}{P^{i_0}_0(h^{i_0}_0)^2+K^{i_0}_0}x^2,
\end{align*}
where $\lambda^*$ is given by \eqref{lambdastar}, and the maximum value is exactly the right-hand side of \eqref{efficientfron}. Finally, the optimal portfolio \eqref{efficient} is obtained by setting $\lambda=\lambda^*$ in \eqref{pistar1}.
\end{proof}

\section{MV portfolio selection with no-shorting constraint}\label{sec:noshort}
In this section, we study MV portfolio selection problem \eqref{optm} with no-shorting constraint, i.e. $\Pi=\R^m_+$.

First for any $v\in\R^m_+$ and $P_+=(P^1_+,\ldots,P^{\ell}_+)^{\top}\in\R^{\ell}_{++}, \ P_-=(P^1_-,\ldots,P^{\ell}_-)^{\top}\in\R^{\ell}_{++}$, we define the following mappings:
\begin{align*} H^i_{+,t}(v,P_+,P_-) 
:=&~P_+^i|v^{\top}\sigma^i_t|^2+2P_+^i(v^{\top}\mu^i_t-\int_{\cE}v^{\top}\beta^i_t\nu(\de))\medskip\\
&~+\sum_{k=1}^{n_2}\int_{\cE}\Big[P_+^i[(1+v^{\top}\beta_{k,t}^i)^+]^2-P_+^i+P_-^i[(1+v^{\top}\beta_{k,t}^i)^-]^2\Big]\nu_k(\de)\medskip\\
&~+\sum_{j\neq i}^{\ell}q^{ij}\Big[P_+^j [(1+v^{\top}\gamma^{i,j}_t)^+]^2-P_+^i+P_-^j[(1+v^{\top}\gamma^{i,j}_t)^-]^2\Big], \medskip\\
H^i_{-,t}(v,P_+,P_-) 
:=&~P_-^i|v^{\top}\sigma^i_t|^2-2P_-^i(v^{\top}\mu^i_t-\int_{\cE}v^{\top}\beta^i_t\nu(\de))\medskip\\
&~+\sum_{k=1}^{n_2}\int_{\cE}\Big[P_-^i[(1-v^{\top}\beta_{k,t}^i)^+]^2-P_-^i+P_+^i[(1-v^{\top}\beta_{k,t}^i)^-]^2\Big]\nu_k(\de)\medskip\\
&~+\sum_{j\neq i}^{\ell}q^{ij}\Big[P_-^j [(1-v^{\top}\gamma^{i,j}_t)^+]^2-P_-^i+P_+^j[(1-v^{\top}\gamma^{i,j}_t)^-]^2\Big],
\end{align*}
and
\begin{align}
H^{*,i}_{+,t}(P_+,P_-)&:=\inf_{v\in\R^m_+}H^i_{+,t}(v,P_+,P_-),\label{defH1}\medskip\\
H^{*,i}_{-,t}(P_+,P_-)&:=\inf_{v\in\R^m_+}H^i_{-,t}(v,P_+,P_-).\label{defH2}
\end{align}

\begin{remark}\label{Hfinite}
Under Assumption \ref{assu1}, both $H^i_{+,t}(P_+,P_-)$ and $H^i_{-,t}(P_+,P_-)$ take finite values for any $P_+, P_-\in\R^{\ell}_{++}$. In fact, for any $(P_{+},P_{-})\in \R^{\ell}_{++}\times \R^{\ell}_{++}$, there are constants $c_1(P_{+},P_{-})>0$ and $c_2(P_{+},P_{-})>0$ such that
\begin{equation*}
H^i_{\pm}(v,P_{+},P_{-})-H^i_{\pm}(0,P_{+},P_{-})\geq c_1|v|^2-c_2(|v|+1).
\end{equation*}
Hence $H^i_{\pm}(v,P_{+},P_{-})> H^i_{\pm}(0,P_{+},P_{-})\geq H_{\pm}^{*,i}(P_{+},P_{-})$ whenever $|v|$ is sufficiently large in terms of $c_1$ and $c_2$.
Therefore, we have
\begin{equation}\label{bound}
H^{*,i}_{\pm,t}(P_+,P_-)=\inf_{v\in\R^m_+,|v|\leq c_3(P_+,P_-)}H^i_{\pm,t}(v,P_+,P_-),
\end{equation}
where $c_3(P_+,P_-)>0$. It follows that
$H_{\pm}^*(P_{+},P_{-})$ are finite and locally Lipschitz w.r.t. $(P_{+},P_{-})$.
\end{remark}

Now we introduce the following Riccati equation, which is a nonlinear $2\ell$-dimensional fully coupled ODE system:
\begin{align}
\label{Riccati}
\begin{cases}
\dot{P^i_{+,t}}
=-\Big[2r^i_tP^i_{+,t}+H^{*,i}_{+,t}(P_{+,t},P_{-,t})\Big]\medskip\\
\dot{P^i_{-,t}}
=-\Big[2r^i_tP^i_{-,t}+H^{*,i}_{-,t}(P_{+,t},P_{-,t})\Big]\medskip\\
P^i_{+,T}=P^i_{-,T}=1, \medskip\\
P^i_{+,t}>0, \ P^i_{-,t}>0, \ t\in[0,T], \ i\in\cM.
\end{cases}
\end{align}
The solution $(P_+,P_-)=(P^1_+,\ldots,P^{\ell}_+,P^1_-,\ldots,P^{\ell}_-)$ of \eqref{Riccati} is required to be positive in the last constraints, hence $H^{*,i}_{+,t}$ and $H^{*,i}_{-,t}$ are well defined as delineated in Remark \ref{Hfinite}.

\begin{thm}\label{existence}
Let Assumption \ref{assu1} hold, and let $H_{+}^*$ and $H_{-}^*$ be defined by \eqref{defH1} and \eqref{defH2}. Then the Riccati equation \eqref{Riccati} admits a unique positive solution $(P_{+},P_{-})$. Furthermore,
\begin{equation*}
\hat v^i_{\pm,t}(P_{+,t},P_{-,t}):=\mathop{\mathrm{argmin}}\limits_{v\in\R^m_+}H^i_{\pm,t}(v,P_{+,t},P_{-,t})
\end{equation*}
are bounded functions on $[0,T]$ for all $i\in\cM$.
\end{thm}
\begin{remark}\label{special}
If the infimums in \eqref{defH1} and \eqref{defH2} are taken over $v\in\R^m$, then we find that
$H^{*,i}_{+,t}(P_+,P_-)=H^{*,i}_{-,t}(P_+,P_-)$, hence $P_+=P_-$. And
\begin{align*}
&~\inf_{v\in\R^m}\Big[P^{i}_t|v^{\top}\sigma^{i}_t|^2
+P^i_t\sum_{j=1}^{n_2}\int_{\cE}|v^{\top}\beta^i_{j,t}(e)|^2\nu_{j}(\de) +2P^{i}_tv^{\top}\mu^{i}_t+\sum_{j\neq i}^{\ell}q^{ij}\Big(P^{j}_t(1+v^{\top}\gamma^{i,j}_t)^2-P^{i}_t\Big)\Big]\medskip\\
=&~ -\mathscr{M}_t^i(P)^{\top}(\mathscr{N}_t^i(P))^{-1}\mathscr{M}_t^i(P) +\sum\limits_{j\neq i}q^{ij}(P^j_t-P^i_t),
\end{align*}
so \eqref{Riccati} degenerates to \eqref{P}.
\end{remark}

\begin{proof}
\textbf{Step 1: Truncation.}
Assumption \ref{assu1} implies there are positive constants $c_1$ and $c_2$ such that for any $ t\in[0,T], \ i\in\cM$,
\begin{equation}\label{leq}
r^i_t\leq c_1,
\end{equation}
and
\begin{equation}\label{geq}
2r^i_t-(\mu^i_t+\sum_{j\neq i}^{\ell}q^{ij}\gamma^{ij})^{\top}(\Sigma^i_t)^{-1}(\mu^i_t+\sum_{j\neq i}^{\ell}q^{ij}\gamma^{ij})\geq -c_2.
\end{equation}

We first introduce some notations. Let $\kappa:=e^{2c_1T}$ and $\epsilon:=e^{-c_2T}$.
For $P=(P^1,...,P^{\ell})^{\top}\in\R^{\ell}$, let $g(P)=(\epsilon\vee(P^1\wedge \kappa),...,\epsilon\vee(P^{\ell}\wedge \kappa))\in\R^{\ell}$. Then $H_{\pm,t}^{i,*}(g(P_{+}),g(P_{-}))$ are Lipschitz continuous w.r.t. $(P_{+},P_{-})$ on $\R^{\ell}\times\R^{\ell}$, so the $2\ell$-dimensional ODE
\begin{align}\label{Riccatitrun}
\begin{cases}
\dot{\tilde P}^i_{+,t}=-\Big[2r^i_t\tilde P^i_{+,t}+H_{+,t}^{*,i}(g(\tilde P_{+,t}),g(\tilde P_{-,t}))\Big],\medskip\\
\dot{\tilde P}^i_{-,t}=-\Big[2r^i_t\tilde P^i_{-,t}+H_{-,t}^{*,i}(g(\tilde P_{+,t}),g(\tilde P_{-,t}))\Big],\medskip\\
\tilde P^i_{\pm,T}=1,
\end{cases}
\end{align}
admits a unique classical solution, denoted by $(\tilde P_{+},\tilde P_{-})$. In the next two steps, we will respectively show
\begin{align*}
\tilde P^i_{\pm,t}\leq \kappa, \ t\in[0,T], \ i\in\cM,
\end{align*}
and
\begin{align}\label{lower}
\tilde P^i_{\pm,t}\geq \epsilon, \ t\in[0,T], \ i\in\cM.
\end{align}
Then $(\tilde P_{+},\tilde P_{-})$ is actually a positive solution to the ODE \eqref{Riccati} since $g(\tilde P_{\pm,t})=\tilde P_{\pm,t}$ in \eqref{Riccatitrun}.

\textbf{Step 2: Upper bound.}
Notice that $\overline P_t:=e^{2c_1(T-t)}$ is the unique solution to the following linear ODE
\begin{align*}
\begin{cases}
\dot P_t=-2c_1P_t,\medskip\\
P_T=1.
\end{cases}
\end{align*}
We shall use $c$ to represent a generic positive constant which can be different from line to line.
According to the chain rule, we have
\begin{align*} 
[(\tilde P^i_{+,t}-\overline P_t)^+]^2 
&=\int_t^T 2(\tilde P^i_{+,s}-\overline P_s)^+\big[2r^i_s(\tilde P^i_{+,s}-\overline P_s)+2(r^i_s-c_1)\overline P_s +H^{*,i}_{+,s}(\tilde P_{+,s},\tilde P_{-,s})\big]\ds\medskip\\
&\leq c\int_t^T (\tilde P^i_{+,s}-\overline P_s)^+\big[|\tilde P^i_{+,s}-\overline P_s|+H^{i}_{+,s}(0,\tilde P_{+,s},\tilde P_{-,s})\big] \ds\medskip\\
&= c\int_t^T (\tilde P^i_{+,s}-\overline P_s)^+\big[|\tilde P^i_{+,s}-\overline P_s|+\sum_{j\neq i}^{\ell}q^{ij}(\tilde P^j_{+,s}-\tilde P^i_{+,s}) -\sum_{j\neq i}^{\ell}q^{ij}(\overline P_{s}-\overline P_{s})\big]\ds\medskip\\
&\leq c\sum_{j=1}^{\ell}\int_t^T [(\tilde P^j_{+,s}-\overline P_s)^+]^2\ds,
\end{align*}
where we used \eqref{leq} in the first inequality, and the fact $q^{ij}\geq 0$ for $i\neq j$ in the second inequality. Taking sums in terms of $i$ from $1$ to $\ell$ on both sides of the above inequality, we get
\begin{equation*}
\sum_{i=1}^{\ell}[(\tilde P^i_{+,t}-\overline P_t)^+]^2
\leq c\ell\sum_{i=1}^{\ell}\int_t^T [(\tilde P^i_{+,s}-\overline P_s)^+]^2\ds.
\end{equation*}
Gronwall's inequality implies $\sum_{i=1}^{\ell}[(\tilde P^i_{+,t}-\overline P_t)^+]^2=0$, so 
\begin{equation*}
\tilde P^i_{+,t}\leq \overline P_t= e^{2c_1(T-t)}\leq \kappa.
\end{equation*}
The proof for the upper bound of $\tilde P^i_{-,t}$ is parallel, so we omit it.

\textbf{Step 3: Lower bound.}
To prove the estimate \eqref{lower},
we notice $\underline P_t=\exp(-c_2(T-t))$, $t\in[0,T]$ satisfies the following ODE
\begin{align*}
\begin{cases}
\dot{P}_t=c_2 P_{t},\medskip\\
P_T=1.
\end{cases}
\end{align*}
Clearly, $\epsilon\leq\underline P_t\leq \kappa$, so we have, by \eqref{geq} and the definitions of $H^i_{+}$ and $H^{*,i}_{+}$,
\begin{align*} 
-c_2\underline P_t 
&\leq [2r^i_t
-(\mu^i_t+\sum_{j\neq i}^{\ell}q^{ij}\gamma^{ij})^{\top}(\Sigma^i_t)^{-1}(\mu^i_t+\sum_{j\neq i}^{\ell}q^{ij}\gamma^{ij})]\underline P_{t}\medskip\\
&=2r^i_t\underline P_{t}+\inf_{v\in\R^m}H^i_{+,t}(v,\underline{\mathbf{P}}_t,\underline{\mathbf{P}}_t)\medskip\\
&\leq 2r^i_t\underline P_{t}+ \inf_{v\in\R^m_+}H^i_{+,t}(v,\underline{\mathbf{P}}_t,\underline{\mathbf{P}}_t)\medskip\\
&=2r^i_t\underline P_{t}+H_{+,t}^{*,i}(g(\underline{\mathbf{P}}_t),g(\underline{\mathbf{P}}_t)).
\end{align*}
where $\underline{\mathbf{P}}_t$ stands for the $\ell$-dimensional column vector whose components are all $\underline P_t$.
It follows from the chain rule that
\begin{align*}\label{chain}
[(\underline P_t-\tilde P^i_{+,t})^+]^2&=\int_t^T 2(\underline P_s-\tilde P^i_{+,s})^+\Big[-c_2\underline P_s-2r^i_t\tilde P_{+,s}-H_{+,s}^{*,i}(g(\tilde P_{+,s}),g(\tilde P_{-,s}))\Big]\ds\nn\medskip\\
&\leq \int_t^T 2(\underline P_s-\tilde P^i_{+,s})^+\Big[2r^i_s(\underline P_s-\tilde P^i_{+,s})\medskip\\
&\quad\;+H_{+,s}^{*,i}(g(\underline{\mathbf{P}}_t),g(\underline{\mathbf{P}}_t))
-H_{+,s}^{*,i}(g(\tilde P_{+,s}),g(\tilde P_{-,s}))\Big]\ds.
\end{align*}
On the other hand, from \eqref{bound}, we know
\begin{equation*}\label{bounddomain}
H_{+,t}^{*,i}(g(\tilde P_{+,t}),g(\tilde P_{-,t}))=\inf_{v\in\R^m_+,|v|\leq c}H^i_{+,t}(v,g(\tilde P_{+,t}),g(\tilde P_{-,t})).
\end{equation*}
Observe that $H^i_{+,t}(v,g(P_{+}),g(P_{-}))$ is non-decreasing w.r.t. $P^j_{+}$ for $j\neq i$ and every component $P^i_{-}$ of $P_{-}$, $i\in\cM$, we get
\begin{align*}
[(\underline P_t-\tilde P^i_{+,t})^+]^2
&\leq \int_t^T 2(\underline P_s-\tilde P^i_{+,s})^+\Big[2r^i_s(\underline P_s-\tilde P^i_{+,s})\medskip\\
&\qquad+\sup_{v\in\R^m_+,|v|\leq c}\Big(H^i_{+,s}(v,g(\underline{\mathbf{P}}_t),g(\underline{\mathbf{P}}_t)) -H^i_{+,s}(v,g(\tilde P_{+,s}),g(\tilde P_{-,s}))\Big)\Big]\ds\medskip\\
&\leq c\int_t^T (\underline P_s-\tilde P^i_{+,s})^+\Big[|\underline P_s-\tilde P^i_{+,s}| +\sum_{j\neq i}^{}(\underline P_s-\tilde P^j_{+,s})^++\sum_{j=1}^{\ell}(\underline P_s-\tilde P^j_{-,s})^+\Big]\ds\nn\medskip\\
&\leq c\sum_{i=1}^{\ell}\int_t^T\Big[[(\underline P_s-\tilde P^i_{+,s})^+]^2+[(\underline P_s-\tilde P^i_{-,s})^+]^2\Big]\ds.
\end{align*}
Similarly, we have
\begin{equation*}
[(\underline P_t-\tilde P^i_{-,t})^+]^2
\leq c\sum_{i=1}^{\ell}\int_t^T[(\underline P_s-\tilde P^i_{+,s})^+]^2+[(\underline P_s-\tilde P^i_{-,s})^+]^2\ds.
\end{equation*}
Combining the above two inequalities and taking sums in terms of $i$ from $1$ to $\ell$ on both sides yield
\begin{align*}
\sum_{i=1}^{\ell}\Big[[(\underline P_t-\tilde P^i_{+,t})^+]^2+[(\underline P_t-\tilde P^i_{-,t})^+]^2\Big] &\leq 2c\ell\sum_{i=1}^{\ell}\int_t^T[(\underline P_s-\tilde P^i_{+,s})^+]^2+[(\underline P_s-\tilde P^i_{-,s})^+]^2\ds.
\end{align*}
We infer from Gronwall's inequality that $[(\underline P_t-\tilde P^i_{+,t})^+]^2+[(\underline P_t-\tilde P^i_{-,t})^+]^2=0$, which implies $\tilde P^i_{\pm,t}\geq \underline P_t=e^{-c_2(T-t)}\geq \epsilon$, establishing \eqref{lower}.

\textbf{Step 4: Uniqueness.}

Suppose \eqref{Riccati} admits two solutions $(P_{+}, P_{-})$ and $(P'_{+}, P'_{-})$ which are continuous obviously. Noting the last constraints in \eqref{Riccati}, we conclude that there exists two positive constants $\kappa'\geq\epsilon'$, such that
\[
\epsilon'\leq P^{i}_{\pm,t}\leq \kappa', \ \epsilon'\leq P^{'i}_{\pm,t}\leq \kappa', \ t\in[0,T], \ i\in\cM.
\]

Replacing $\kappa$ and $\epsilon $ in the truncation function $g$ by $\kappa\vee\kappa'$ and $\epsilon\wedge\epsilon'$ in \eqref{Riccatitrun}, then both $(P_{+}, P_{-})$ and $(P'_{+}, P'_{-})$ are solutions to the ODE \eqref{Riccatitrun} with Lipschitz driver, so they must be equal.

\textbf{Step 5.}
Finally, the existence and boundedness of $\hat v_{\pm,t}(P_{+,t},P_{-,t})$ follow from \eqref{bound} and
the boundedness of $P_{\pm,t}$.
\end{proof}

To solve the MV problem \eqref{optm} or LQ problem \eqref{optmun} under no-shorting constraints, hereafter we have to impose the following assumption.

\begin{assmp}\label{assu2}
The interest rate $r_t$ is a deterministic, Borel-measurable and bounded function on $[0,T]$.
\end{assmp}
Under Assumption \ref{assu2}, the interest rate $r_t$ is independent of the Markov chain $\alpha$, and the unique solutions to the $\ell$-dimensional linear ODEs \eqref{psi} and \eqref{h} are given by
\[
\psi^i_t=e^{\int_t^Tr_s\ds}, \ h^i_t=e^{-\int_t^Tr_s\ds}, \ t\in[0,T], \ i\in\cM.
\]
Since they do not depend on $i\in\cM$, we henceforth write $\psi_t$ and $h_t$ respectively.
And the feasible condition \eqref{assufeas2} is equivalent to
\begin{align}\label{assufeas3}
\sum_{k=1}^{\ell}\E\int_0^T\Big(\mu^{\al}_{k,t}+\sum_{j\neq\al}q^{\al j}\gamma^{\al,j}_{k,t}\Big)^+\dt>0.
\end{align}

The following result may be proved in much the same way as \cite[Proposition 4.1]{SX}. Details are left to the interested readers.
\begin{thm}\label{verifi}
Suppose Assumptions \ref{assu1} and \ref{assu2} hold. Also suppose the feasible condition \eqref{assufeas2} (or equivalently $\eqref{assufeas3}$) holds.
Let $(P_{+},P_{-})$, $\hat v_{\pm,t}(P_{+,t},P_{-,t})$ be given in Theorem \ref{existence}.
Then the LQ problem \eqref{optmun} admits an optimal feedback control as follows
\begin{align}\label{pistar2}
\pi^*(t,X)&=\hat v^{\al}_{+,t}(P_{+,t},P_{-,t})(X_{t-}-\lambda h_t)^+ +\hat v^{\al}_{-,t}(P_{+,t},P_{-,t})(X_{t-}-\lambda h_t)^-.
\end{align}
Moreover, the corresponding optimal value is
\begin{align*}\label{value2}
V(x,i_0;\lambda) 
=P^{i_0}_{+,0}\big[(x-\lambda h_0)^+\big]^2+P^{i_0}_{-,0}\big[(x-\lambda h_0)^-\big]^2.
\end{align*}
\end{thm}

Under Assumption \ref{assu2}, one can verify that $\pi=0$ is an optimal portfolio for the MV problem \eqref{optm} when $z=xh_0^{-1}$. The case $z<xh_0^{-1}$ is financially meaningless as any rational investor would expect a better return than doing nothing, hence, we will mainly focus on the nontrivial case $z>xh_0^{-1}$ in this case.

Next we proceed to derive the optimal portfolio and efficient frontier for the original MV problem \eqref{optm}.
\begin{thm}\label{Th:efficient2}
Under the conditions of Theorem \ref{verifi} and suppose the feasible condition \eqref{assufeas2} (or equivalently $\eqref{assufeas3}$) {holds}.
Then we have
\begin{align}\label{nonnegative2}
h_0^2P^{i_0}_{-,0}<1.
\end{align}
Moreover, the state feedback control defined by
\begin{align}\label{efficient2}
\pi^*(t,X)&=\hat v^{\al}_{+,t}(P_{+,t},P_{-,t}) (X_{t-}-\lambda^*h_t)^+ +\hat v^{\al}_{-,t}(P_{+,t},P_{-,t})(X_{t-}-\lambda^*h_t)^-,
\end{align}
is optimal for the MV problem \eqref{optm}, where
\begin{align}\label{lambdastar2}
\lambda^*=\frac{z-xh_0P^{i_0}_{-,0}}{1-h_0^2P^{i_0}_{-,0}}.
\end{align}
Furthermore, the efficient frontier to the problem \eqref{optm} is a half-line, determined by
\begin{align}
\label{efficientfron2}
\mathrm{Var}(X^{\pi^{*}}_T)=\frac{P^{i_0}_{-,0}}{1-h^2_0 P^{i_0}_{-,0}}\Big(h_0\E[X_T^{\pi^{*}}]-x\Big)^2,
\end{align}
where $\E[X_T^{\pi^{*}}]\geq xh_0^{-1}$.
\end{thm}
\begin{proof}
Recalling $h_t=e^{-\int_t^Tr_s\ds}$.

\textbf{Step 1: To prove \eqref{nonnegative2}.}
Applying It\^o's formula to $P^{\alpha_t}_{-,t}h_t^2$, we get
\begin{align}\label{Phsqu}
\dd \ (P^{\alpha_t}_{-,t}h_t^2)&= h_t^2\Big[-H_{-,t}^{*,\al}(P_{+},P_{-}) +\sum_{j\neq \al}^{\ell}q^{\al j}(P^j_--P^{\al}_-)\Big]\dt\nn\medskip\\
&\quad\;+h_t^2\sum_{j\neq \al}^{\ell}(P^j_--P^{\al}_-)\dd \tilde N^{\al,j}_t.
\end{align}
Integrating from $0$ to $T$, and taking expectation, we get
\begin{align*}
1-P^{i_0}_{-,0}h_0^2=\E\int_0^T h_t^2\Big[&-H_{-,t}^{*,\al}(P_{+},P_{-}) +\sum_{j\neq \al}^{\ell}q^{\al j}(P^j_--P^{\al}_-)\Big]\dt.
\end{align*}
On the other hand, it is clear from the definition of $H_{-,t}^{*,i}$ in \eqref{defH1} that
\begin{equation*}
H_{-,t}^{*,i}(P_{+},P_{-})\leq H_{-,t}^i(0,P_{+},P_{-})=\sum_{j\neq i}^{\ell}q^{ij}(P^j_--P^i_-).
\end{equation*}
The above two inequalities lead to
\[
1-P^{i_0}_{-,0}h_0^2\geq 0,
\]
and the equality holds only if
\begin{equation}\label{contradiction2}
H_{-,t}^{*,\al}(P_{+},P_{-})= \sum_{j\neq \al}^{\ell}q^{\al j}(P^j_--P^{\al}_-).
\end{equation}
Suppose that \eqref{contradiction2} holds. It then follows from \eqref{Phsqu} that $P^{\alpha_t}_{-,t}h_t^2$ is a martingale with terminal value $1$, hence
\begin{equation*}
P^{\alpha_t}_{-,t}h_t^2\equiv1, \ t\in[0,T].
\end{equation*}
It turns out that the diffusion term in \eqref{Phsqu} must equal $0$, resulting in
\begin{equation*}
P^{j}_{-,t}h_t^2=P^{\alpha_t}_{-,t}h_t^2\equiv1, \ t\in[0,T], \ j\in\cM.
\end{equation*}
Furthermore, noting \eqref{contradiction2}, the above equality implies
\begin{equation}\label{contradiction3}
H_{-,t}^{*,\al}(P_{+},P_{-})= 0, \ t\in[0,T].
\end{equation}
Since $\beta, \gamma$ are bounded, according to Assumption \ref{assu1}, we can choose a sufficiently small constant $\varepsilon>0$, which will be fixed in the following proof, such that
$1-v^{\top}\beta^i_k>0$ and $1-v^{\top}\gamma^{i,j}> 0$ 
for all $i,j\in\cM$, $k=1,\ldots,n_2$, $v\in\R^m_+$ with $|v|\leq \varepsilon$. 
For any $v\in\R^m_+$ with $|v|\leq \varepsilon$, we have (recalling the definition of $\Sigma^i$ in Assumption \ref{assu1})
\begin{align*}
h^2_tH^{\al}_{-,t}(v,P_{+,t},P_{-,t})
&=|v^{\top}\sigma^{\al}_t|^2-2\Big(v^{\top}\mu^{\al}_t-\int_{\cE}v^{\top}\beta^{\al}_t\nu(\de)\Big)\medskip\\
&\quad\;+\sum_{k=1}^{n_2}\int_{\cE}\Big[(1-v^{\top}\beta_{k,t}^{\al})^2-1\Big]\nu_k(\de)\medskip\\
&\quad\;+\sum_{j\neq \al}^{\ell}q^{\al j}\Big[ (1-v^{\top}\gamma^{\al,j}_t)^2-1\Big] \medskip\\
&=v^{\top} \Sigma_t^{\al}v-2v^{\top}(\mu^{\al}_t+\sum_{j\neq\al}q^{\al j}\gamma^{\al,j}_t).
\end{align*}
It is easily seen that the above quadratic function will take strict negative values for proper $v\in\R^m_+$ on a set of $(t,\omega)$ with positive measure under the feasible condition \eqref{assufeas2} or \eqref{assufeas3}. This contradicts \eqref{contradiction3}.
So we proved \eqref{nonnegative2}.

\textbf{Step 2:}
According to Theorem \ref{verifi},
\begin{align*}
V(x,i_0;\lambda)=P^{i_0}_{+,0}[(x-\lambda h_0)^+]^2+P^{i_0}_{-,0}[(x-\lambda h_0)^-]^2
\end{align*}
is a piecewise quadratic function of $\lambda$. 
{Using $zh_0>x$ and the inequality $h_0^2P^{i_0}_{-,0}<1$ from \eqref{nonnegative2}, it follows $$x-\lambda^*h_0=\frac{x-zh_0}{1-h_0^2P^{i_0}_{-,0}}<0,$$
where $\lambda^*$ is given by \eqref{lambdastar2}.
A simple calculation shows that} 
\begin{align*}
\sup_{\lambda\in\R}[V(x,i_0;\lambda)-(\lambda-z)^{2}]&=V(x,i_0;\lambda^*)-(\lambda^*-z)^{2} =\frac{P^{i_0}_{-,0}}{1-h^2_0 P^{i_0}_{-,0}}(zh_0-x)^2,
\end{align*}
and the the maximum value equals exactly the right-hand side of \eqref{efficientfron2}. Finally, the optimal portfolio \eqref{efficient2} is obtained by \eqref{pistar2} with $\lambda=\lambda^*$.
The proof is complete.
\end{proof}

\section{Conclusions}\label{sec:conclusion}
In this paper, we investigated MV portfolio selection problems under regime-switching-induced stock price shocks with and without no-shorting constraint. We derived the optimal portfolio strategy and the efficient frontier by solving three systems of ODEs, including one multi-dimensional, fully coupled, and highly nonlinear Riccati equation. Unfortunately, due to technical limitations, we can only deal with Markovian models. It is of great interest to solve the non-Markovian counterpart, which would lead to study of highly nonlinear BSDEs.

\end{document}